\documentclass[12pt]{article}
\usepackage[utf8]{inputenc}
\usepackage{subfigure,algorithm,algorithmic}

\usepackage{graphicx}
\usepackage{epstopdf}
\usepackage{hyperref, fullpage}
\usepackage{amssymb,amsmath,amsthm}

\newtheorem{thm}{Theorem}
\newtheorem{lem}[thm]{Lemma}

\newtheorem{fac}[thm]{Fact}

\renewcommand{\leq}{\leqslant}
\renewcommand{\geq}{\geqslant}

\newcommand{\OPT}{\operatorname{OPT}}
\newcommand{\LP}{\operatorname{LP}}
\newcommand{\bigUn}{{1\!\!1}}

\newcommand{\LinProg}[3]{
	\ensuremath{
	\left\{
	\begin{array}{rl}
	\operatorname{#1} 
	& #2\\[5mm]
	\text{such that:} #3
	\end{array}
	\right.
	}
}
\newcommand{\MinProg}[2]{\LinProg{Minimize}{#1}{#2}}

\newcommand{\OMITTED}[1]{}

\title{Facility Location in Evolving Metrics%
	\thanks{This work was partially supported by the ANR-2010-BLAN-0204 Magnum and ANR-12-BS02-005 RDAM grants.}}
	
\author{David Eisenstat\thanks{Brown University (USA)}%
\and Claire Mathieu\thanks{CNRS, \'Ecole normale sup\'erieure UMR 8548 (France) - \href{http://www.di.ens.fr/ClaireMathieu.html}{\texttt{http://www.di.ens.fr/ClaireMathieu.html}}}%
\and Nicolas Schabanel\thanks{CNRS, Universit\'e Paris Diderot (France) - \href{http://www.liafa.univ-paris-diderot.fr/~nschaban/}{\texttt{http://www.liafa.univ-paris-diderot.fr/$\sim$nschaban/}} -
	IXXI, \'Ecole normale sup\'erieure de Lyon (France)}
}

\begin{document}
\maketitle

\begin{abstract}
Understanding the dynamics of evolving social or infrastructure networks is a challenge in applied areas such as epidemiology, viral marketing, or urban planning. During the past decade, data has been collected on such networks but has yet to be fully analyzed. We propose to use information on the dynamics of the data to find stable partitions of the network into groups. For that purpose, we introduce a time-dependent, dynamic version of the facility location problem, that includes a switching cost when a client's assignment changes from one facility to another. This might provide a better representation of an evolving network, emphasizing the abrupt change of relationships between subjects rather than the continuous evolution of the underlying network. We show that  in realistic examples this model yields indeed better fitting solutions than optimizing every snapshot independently.  We present an $O(\log nT)$-approximation algorithm and a matching hardness result, where $n$ is the number of clients and $T$ the number of time steps. We also give an other algorithms with approximation ratio $O(\log nT)$  for the  variant  where one pays at each time step (leasing) for each open facility.
\end{abstract}


\section{Introduction}

During the past decade, a massive amount of data has been collected on diverse networks such as web links, nation- or world-wide social networks, online social networks (Facebook or Twitter for example), social encounters in hospitals, schools, companies, or conferences (e.g.~\cite{Newman2003,StehleVoirinBarrat+2011}), and other real-life networks. Those networks evolve with time, and their dynamics have a considerable impact on their structure and effectiveness (e.g.~\cite{PastorSatorrasVespignani2001,Kleinberg2004SIAMNews}). Understanding the dynamics of evolving networks is a central question in many applied areas such as epidemiology, vaccination planning, anti-virus design, management of human resources, viral marketing, ``facebooking'', etc.  Obtaining a relevant clustering of the data is often a key to the design informative representations of massive data sets. Algorithmic approaches have for instance been successful in yielding useful insights on several real networks such as zebras social interaction networks~\cite{TantipathananandhBergerWolfKempe2007}.

But the dynamics of real-life evolving networks are not yet well understood, partly because it is difficult to observe and analyze such large networks sparsely connected over time. Some basic facts have been observed (such as the preferential attachment  or copy-paste mechanisms) but more specific structures remain to be discovered.  In this article, we propose a new formulation of the facility location problem adapted to these evolving networks. We show that requiring the solution to be stable over time yields in many realistic situations better fitting solutions than optimizing independently various snapshots of the network. 

\paragraph{The problem.}  We focus on the facility location problem where clients are moving in some space: we look for the best connections of clients to facilities (sometimes called centers) over time minimizing a tradeoff between three objectives. The two first objectives are classically: the \emph{distance cost} of the connections (the sum of the connection lengths), so that each client gets connected to a facility representative of its position; and the \emph{opening cost}, a price paid for opening each open facility, so that only the most meaningful facilities (and as few of them as possible) get open. The third and new objective is the \emph{instability} of the connections over time, measured as the number of clients switching from one facility to another over time, so that only the changes responding to significative and lasting changes in the metrics get authorized. We argue that incorporating this stability requirement in the objective function helps in many realistic situations to obtain more desirable solutions (see Section~\ref{sec:def}).

\paragraph{Related work.} Facility location problem has been studied extensively in the offline, online and incremental settings, see \cite{Fotakis2011} for a survey. The offline version of the problem was a case study accompanying the development of techniques for approximation algorithms: primal-dual and dual fitting methods and local search for example. A series of papers, \cite{ShmoysTardosAardal1997,MahdianYeZhang2002,JainMahdianMarkakisSaberiVazirani2003,AryaGargKhandekarMeyersonMunagalaPandit2004,CharikarGuha2005,ByrkaAardal2010,Li2011}, obtained  (almost) matching upper and lower bounds on the polynomially achievable approximation ratio in this setting: $\Theta(\log n)$ in the non-metric case, and within $[1.463,1.488]$ in the metric case, when the client-to-facility connection cost is a distance in a metric space. 

The online setting, where clients arrive over time and the algorithm gradually buys more and more facilities to serve them, was first addressed by \cite{Meyerson2001} who obtained $\Theta(\log n/\log\log n)$ upper and lower bounds on the competitive ratio of any online algorithm. This  later led to developments for various cases, e.g. analyzed when clients are drawn from some distribution~\cite{AnagnostopoulosBentUpfalVanHentenryck2004} and in other cases~\cite{Fotakis2008}. In order to allow more flexibility in the solution (as required in many clustering application),  incremental approaches, which allow reconsidering the assignment of clients to facilities over time, were also considered. Such variants may allow better ($O(1)$) competitive ratios, see e.g. in the metric case but with streaming constraints~\cite{Fotakis2006}, and  in the special case in the Euclidian setting when facilities may be moved as new clients arrive~\cite{DivekiImreh2011}. We also mention the related clustering problem in which clusters may be merged but not split, e.g.~\cite{CharikarChekuriFederMotwani1997}.

Our approach differs from the existing algorithmic approaches to dynamic settings because we focus on settings where distances may vary over time, and where it is desirable to achieve a tradeoff between the \emph{stability} of the solution --- clusters of clients tend not to flip-flop constantly ---  and its \emph{adaptability} --- the assignment ought to be modified if distance change too much. We show that offline static algorithms that construct an independent optimal solution for each snapshot of the network yield results that, in a large variety of realistic situations, are not only unstable (and thus arbitrarily bad for our objective), but also undesirable with respect to network dynamics analysis. Online solutions such as~the clustering of Charikar et al.~\cite{CharikarChekuriFederMotwani1997} are also unnecessarily pessimistic in this setting: we have access to the whole evolution of the network over time (as given by experiments such as~\cite{StehleVoirinBarrat+2011}) and we can thus anticipate future changes. 


As far as we know, the case where the distance between points vary overtime is still largely unexplored.

\paragraph{Our results.} 

After defining the problem formally in section~\ref{sec:def} and giving examples showing the benefits one can expect from solving this problem in the context of metrics evolving with time, we give in Section~\ref{sec:fixed:alg} a $O(\log nT)$-approximation algorithm for this problem, where $n$ is the number of clients and $T$ the number of time steps.

\begin{thm}[Fixed opening cost] \label{thm:fixed:alg}
There is a polynomial time randomized algorithm which outputs a solution to the dynamic facility location problem with fixed opening cost whose cost verifies: 
$$\Pr\bigl\{cost \leq 4\log(2nT) \cdot \OPT\bigr\} \geq \Pr\bigl\{cost \leq 4\log(2nT) \cdot \LP\bigr\} \geq 1/4.$$
\end{thm}

Repeating this algorithm $O(\log \frac1\epsilon)$ times and outputing the best solution increases the success probability to $1-\epsilon$ for arbitrarily small $\epsilon>0$. We then show in Section~\ref{sec:fixed:hard} that this approximation ratio is asymptotically optimal, even if one assume that the distance verifies the triangle inequality at every time step and if the input consists in only one client and two possible positions for the client and the facilities. 

\begin{thm}[Hardness for fixed opening cost] \label{thm:fixed:hard}
Unless $P\neq N\!P$, there is no $o(\log T)$-approximation, even for the metric case with one client and two possible positions.
\end{thm}

This new problem differs then significantly from the classic facility location problem which  admits no $o(\log n)$-approximation for non-metric distances but a $1.488$-approximation when the distance satisfy the triangle inequality \cite{Li2011}. We then show  in Section~\ref{sec:hourly} how to extend our approximation algorithm to the setting where facilities can be open and closed at any time step and where one pays $f$ for each facility open at each time step.

\begin{thm}[Hourly opening cost] \label{thm:hourly:alg}
There is a polynomial time randomized algorithm which outputs a solution to the dynamic facility location problem with hourly opening cost whose cost verifies: 
$$\Pr\bigl\{cost \leq 4\log(2nT) \cdot \OPT\bigr\} \geq \Pr\bigl\{cost \leq 4\log(2nT) \cdot \LP\bigr\} \geq 1/4.$$
\end{thm}



This article concludes with several open questions and possible extension of this work.

\section{Facility Location in Evolving Metrics}
\label{sec:def}

\subsection{Definition}
\label{sec:def}

\paragraph{Dynamic Facility Location problem with fixed opening cost.} 
We are given a set $F$ of $m$ \emph{facilities} and a set $C$ of $n$ \emph{clients} together with a finite sequence of distances $(d_t)_{1\leq t \leq T}$ over $F\times C$, a non-negative \emph{facility opening cost} $f$ and a non-negative \emph{client switching cost} $g$. The goal is to output a subset $A\subseteq F$ of facilities and, for each time step $t\in[T]$, an assignment $\phi_t:C\rightarrow A$ of facilities to clients, so as to minimize:
\begin{equation*}
f\cdot \#A
+\sum_{1\leq t\leq T, j\in C} d_t(\phi_t (j),j)
+g\cdot\sum_{1\leq t< T}\sum_{j\in C}\bigUn\{\phi_t(j)\neq\phi_{t+1}(j)\},
\end{equation*}
 that is to say the sum of the {opening cost} ($f$ for each open facility), of the total \emph{distance cost} to connect each client to its assigned facility at every time step, and of the \emph{switching cost} for each client ($g$ per change of facility per client). 

\paragraph{Examples.}

\begin{figure}[t]
\center
\subfigure[The classroom: one teacher cycling between 5 groups of students.\label{fig:class}]{
\includegraphics[width=.9\textwidth]{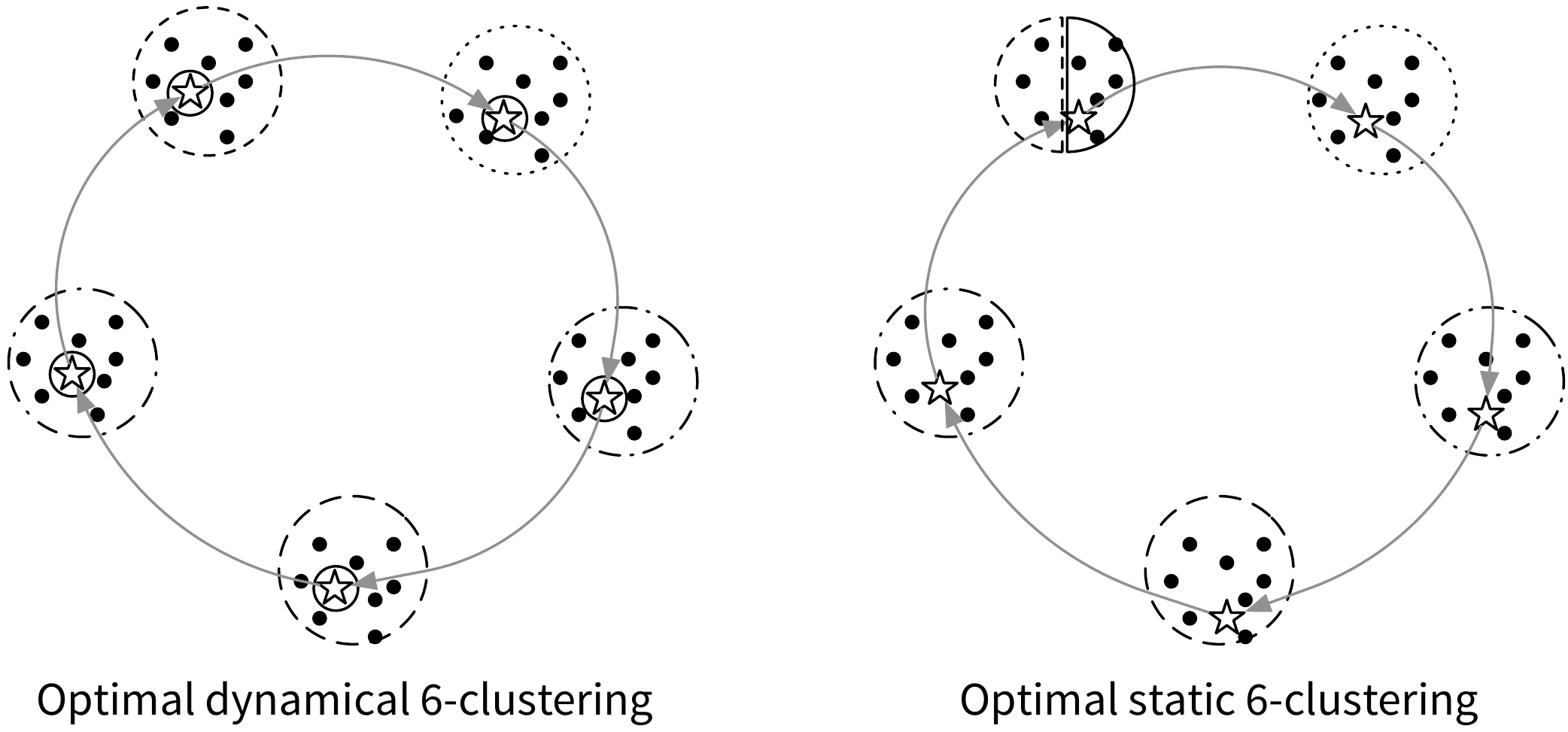}
}
\\[2mm]
\subfigure[Two groups crossing.\label{fig:cross}]{
\includegraphics[width=.9\textwidth]{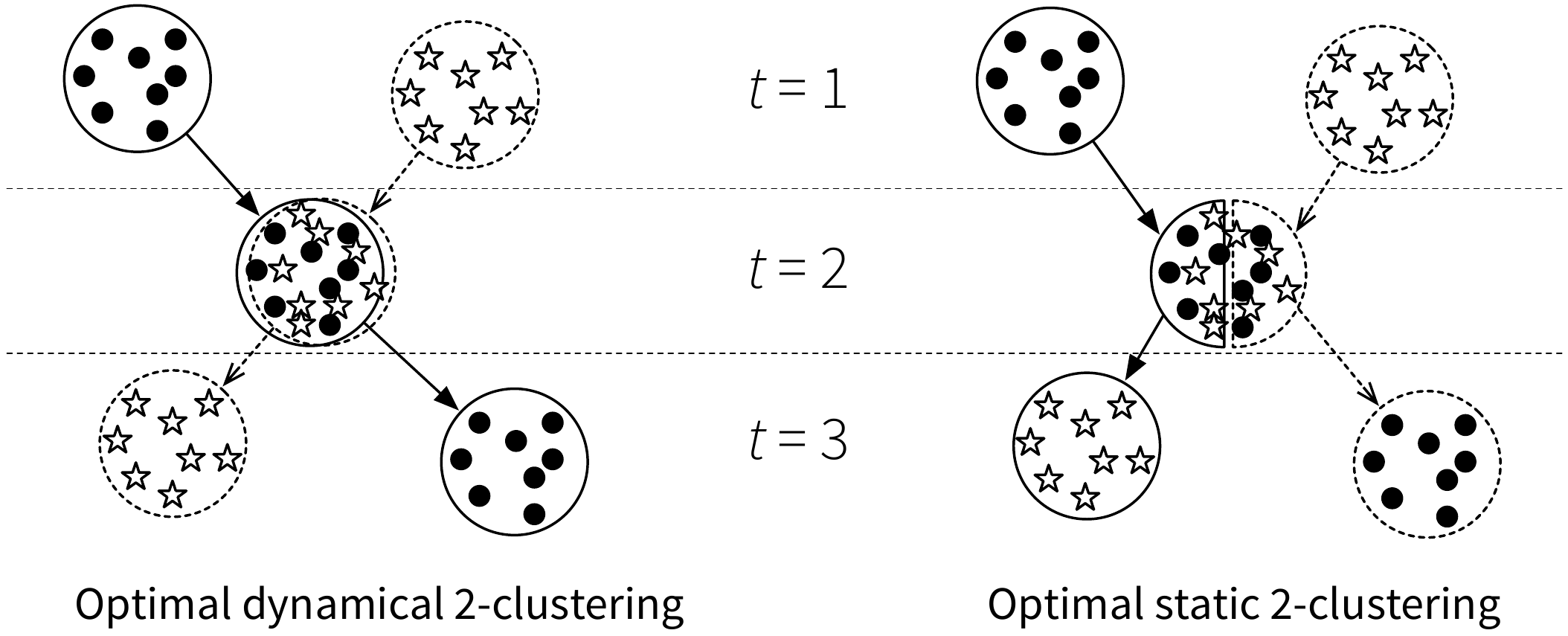}
}
\caption{Dynamic versus static Facility Location.}
\label{fig:statxdyna}
\end{figure}

The two examples in Figure~\ref{fig:statxdyna} show how facility location in the dynamic setting is quite different from the static setting and yields more desirable partitions of the clients. In both examples, a facility can be opened at every client  (so that electing a facility consists in electing a representative for every significantly different behavior). 

In example~\ref{fig:class}, we see a classroom with students split into five groups and a teacher moving from group to group in cyclic order.  When the number of students is large,   static facility location  isolates the five groups and moves the teacher from one group to the next between snapshots; whereas  dynamic facility location isolates every group of students and puts the teacher in a sixth group. 

In example~\ref{fig:cross} we see two groups of people crossing each other (on a street for instance): a static facility location would first output the two groups, then merge them into  a single group, then split it into two groups again; whereas a dynamic facility location would keep the same groups for the whole time period, with the same representatives. 

Assuming in both examples that the distances between  individuals are either very small or very large, then the ratio of the (dynamic) cost  between the dynamic solution and the sequence of static solutions can be made arbitrary large because the switching cost grows for the sequence of static solutions as $\Omega(T)$ and $\Omega(n)$ respectively. 

\begin{fac}
The ratio between the cost of an optimal dynamic facility location solution and the (dynamic) cost of a sequence of optimal static facility location solutions for each snapshot can be as large as $\Omega(T)$ and $\Omega(n)$.
\end{fac}

\paragraph{A linear relaxation.}
For an integer programming formulation, we define indicator 0-1 variables $y_i$, $x_{ij}^t$, and $z_{ij}^t$ for $i\in F$, $j\in C$, and $t\in[T]$: $y_i= 1$ iff facility~$i$ is open; $x^t_{ij}=1$ iff client~$j$ is connected to facility~$i$ at time $t$; and $z^t_{ij} = 1$ iff client~$j$ is connected to facility~$i$ at time~$t$ but no more at time~$t+1$. The dynamic facility location problem is then equivalent to finding an integer solution to the following linear programming relaxation.
\begin{equation}
\MinProg{
	\displaystyle f\cdot \sum_{i\in F} y_i
	+ \sum_{j\in C}\sum_{1\leq t\leq T}\sum_{i\in F} x^t_{ij}\cdot d_t(i,j)
	 + g\cdot \sum_{j\in C} \sum_{1\leq t<T} \sum_{i\in F} z_{ij}^t
	 }{
&	\displaystyle(\forall ijt)~x_{ij}^t\leq y_i\\[2mm]
&	\displaystyle(\forall jt)~\sum_{i\in F} x_{ij}^t=1\\[5mm]
&	\displaystyle(\forall ij,~\forall t< T)~z_{ij}^t\geq x_{ij}^t-x_{ij}^{t+1}\\[2mm]
& \displaystyle y_i,x_{ij}^t,z_{ij}^t\geq 0
} \label{LP:1}
\end{equation}

\subsection{Approximation algorithm}
\label{sec:fixed:alg}


\begin{algorithm}[H]
\caption{Fixed opening cost}\label{alg:fixed:alg}
\begin{algorithmic}
\sffamily
\STATE $\bullet$ Solve the  linear program LP (\ref{LP:1}). Let $(x,y,z)$ be the solution obtained.
\STATE $\bullet$ Draw a facility at random  $\Gamma = 2\log(2nT)\sum_{i\in F} y_i$ times  independently, with distribution proportional to $y$; let $A$ be the resulting multiset of facilities.
\FOR{For each client $j$} 
	\STATE $\bullet$ Determine when it should change from one facility to another using the $z$-variables, and assign it to the cheapest selected facility between each change:
	\STATE ~~~~(a) Partition time greedily  into $\ell_j$ intervals $[t^j_k,t^j_{k+1})$ such that $t^j_1=1$ and where $t^j_{k+1}$ is inductively defined as the largest 
	$t\in(t^j_k,T+1]$ such that ${
\displaystyle\sum_{i\in F} \Bigl(\min_{t^j_k\leq u < t}  x_{ij}^u\Bigr) \geq 1/2}$, and $t^j_{\ell_j+1}= T+1$;
	\STATE~~~~ (b) For each time interval $[t^j_k,t^j_{k+1})$, connect $j$ to the facility in $A$ that is cheapest for $j$ for that time interval.
\ENDFOR
\end{algorithmic}
\end{algorithm}

%
%

Theorem~\ref{thm:fixed:alg}  (page~\pageref{thm:fixed:alg}) states that Algorithm~\ref{alg:fixed:alg} outputs a $O(\log nT)$-approximation with positive constant probability. The next section will show that this is asymptotically optimal (unless ${P\neq N\!P}$).

\begin{proof}[Proof of Theorem~\ref{thm:fixed:alg}]
The expected facility opening cost is obviously at most $2f\log(2nT)\sum_{i\in F} y_i$. In order to bound the switching and distance costs, let us now fix a client~$j$ and show the following fact: 

\begin{fac} \label{fac:sumz>1/2}
For all clients~$j$ and time intervals $[t^j_k,t^j_{k+1})$ with $k<\ell_j$: $$\displaystyle\sum_{t^j_k\leq t <t^j_{k+1}} \sum_{i\in F} z_{ij}^t > 1/2.$$
\end{fac}

This fact yields an easy bound on the switching cost: the switching cost for client~$j$ is $g$ times the number of its intervals minus~$1$. But according to the fact above, for every interval except the last one, the $z_{ij}^t$'s sum to at least $1/2$, so LP~(\ref{LP:1}) pays at least $g/2$ for that interval. The switching cost of the solution is then at most twice the corresponding term in the LP.

\begin{proof}[Proof of Fact~\ref{fac:sumz>1/2}]
For all $t$, $\sum_{i\in F} x_{ij}^t = 1$, in particular at time $t^j_k$. Now, since $k < \ell_j$, we have ${\sum_{i\in F} \Bigl(\min_{t^j_k\leq t\leq t^j_{k+1}}  x_{ij}^{t} \Bigr) < 1/2}$. Let $t_i\in[t^j_k,t^j_{k+1}]$ such that $x_{ij}^{t_i} = \min_{t^j_k\leq t\leq t^j_{k+1}} x_{ij}^t$. We have ${\sum_{i\in F} x_{ij}^{t_i} <1/2}$.  Now, as $z_{ij}^t\geq 0$, ${\sum_{t^j_k\leq t < t^j_{k+1}} z_{ij}^t \geq \sum_{t^j_k\leq t < t_i} z_{ij}^t \geq \sum_{t^j_k\leq t < t_i} (x_{ij}^t-x_{ij}^{t+1}) = x_{ij}^{t^j_k}-x_{ij}^{t_i}}$.  It follows that ${\sum_{t^j_k \leq t <t^j_{k+1}} \sum_{i\in F} z_{ij}^t \geq {\sum_{i\in F} x_{ij}^{t^j_k} - \sum_{i\in F} x_{ij}^{t_i}} > 1-1/2=1/2}$.
\end{proof}

Let us now bound the expected distance cost for client~$j$ within each interval $I=[t^j_k, t^j_{k+1})$. 
Let $x^I_{ij} = \min_{t\in I} x_{ij}^t$ and $\hat{x}^I_{ij}=x^I_{ij}/\sum_{i\in A} x^I_{ij}$.
We want to argue that the facility selection process (which is according to the $y_i$'s) can be simulated, to within a factor of 2, by selecting a facility according to $\hat x^I_{ij}$.  Then the expected distance is correct up to a factor of 2.

We know that $x_{ij}^t\leq y_i$. We can view sampling proportionally to $(y_i)$ as: with probability
$p^I_j=\frac{\sum_i x^I_{ij}}{\sum_i y_i}$, sample proportionally to $x^I_{ij}$, and with the remaining probability, sample proportionally to $y_i-x^I_{ij}$. Indeed,
\begin{align*}
\Pr\{\text{$i$ is selected by this process}\} 
&	= \frac{\sum_i x^I_{ij}}{\sum_i y_i} \cdot \frac{x^I_{ij}}{\sum_i x^I_{ij}} + \left(1-\frac{\sum_i x^I_{ij}}{\sum_i y_i}\right) \cdot \frac{y_i-x^I_{ij}}{\sum_i (y_i-x^I_{ij})} \\
&	= \frac{x^I_{ij}}{\sum_i y_i} + \frac{\sum_i y_i -\sum_i x^I_{ij}}{\sum_i y_i} \cdot \frac{y_i-x^I_{ij}}{\sum_i y_i - \sum_i x^I_{ij}} \\
&	= \frac{y_i}{\sum_i y_i}.
\end{align*}
Formally, we consider the following facility selection process: let $U$ be a uniform real number in $[0,\sum_i y_i)$, we say that facility $i$ is selected if $U\in[\sum_{k<i} y_k, \sum_{k\leq i} y_k)$ and that event $B^I_j$ occurs if $U\in[\sum_{k<i} y_k, \sum_{k< i} y_k+x^I_{ij})$. As pointed out before, according to this process: 1) $i$ is distributed proportionally to $y_i$; 2)~$\Pr B^I_j = p^I_j = \frac{\sum_i x^I_{ij}}{\sum_i y_i} \geq \frac{1}{2\sum_iy_i}$; and 3)~conditioned to event $B^I_j$, $i$ is distributed proportionally to $x^I_{ij}$. 

We repeat the selection process $2\log(2nT)\sum_i y_i$ times independently. Given a pair $(j,I)$, the probability that event $B^I_j$ never occurs is at most $(1-p^I_j)^{2\log(nT)\sum_i y_i} \leq \exp\bigl(-\frac{2\log(2nT)\sum_i y_i}{2\sum_iy_i}\bigr) = \frac{1}{2nT}$. 
Since there are at most $nT$ pairs $(j,I)$, the union bounds ensures that with probability at least $\frac12$, all the events $B^I_j$ occur at least once during the selection.


When $B^I_j$ occurs, the facility $i$ is selected according to $x^I_{ij}$. It follows that the expected distance of this selected facility to $j$ is  for all time $t\in I$:
$\displaystyle \sum_i \frac{x^I_{ij}}{\sum_i x^I_{ij}} \cdot d_t(i,j) \leq \frac{1}{1/2} \sum_i x^t_{ij} d_t(j,i)$.
%
%
It follows that with probability at least $\frac12$, the expectation of the sum of the distances of all $j$'s at all time $t$ to their closest selected facility in $F$ is at most:
$2 \sum_{j,t} \sum_i x^t_{ij} d_t(j,i)$.
%
Summing all the contribution, with probability at least $\frac12$, the expected cost of the solution is at most:
$$
2f\log(2nT)\sum_i y_i
	+ 2  \sum_{i,j,t}x^t_{ij} d_t(j,i)
	+ 2g\sum_{i,j,t} z^t_{ij}
\leq 2\log(2nT) \cdot LP.
$$

Applying Markov inequality we get:
\begin{align*}
\Pr\{cost \leq 4&\log(2nT) LP\}\\ 
&	\geq  \Pr\{cost\leq 2\cdot 2\log(2nT) LP \text{ and all $B_j^I$ occur}\} \\
&	= \Pr\{cost\leq 2\cdot 2\log(2nT) LP ~|~ \text{all $B_j^I$ occur}\} \cdot \Pr\{\text{all $B_j^I$ occur}\}\\
& \geq \frac{1}{2}\cdot\frac12. \quad\text{(Markov)} 
\end{align*}
%
\end{proof}

\subsection{Hardness of approximation}
\label{sec:fixed:hard}

\begin{proof}[Proof of Theorem~\ref{thm:fixed:hard}]
We do a reduction from Set Cover. 

Pick an instance of set cover with $T$ elements and $m$ sets. We define the following instance of dynamic facility location. There is one timestep $t$ for each element of the set cover instance, one facility $i$ for each set of the set cover instance, and a single client.
We set $g=0$ (i.e., $g$ is small enough w.r.t. $f$, $1/n$ and $1/T$). Assume the only possible locations for the client and facilities are two points $a$ and $b$ at distance $\infty$ (i.e. large enough) from each other (note that it satisfies the triangle inequality).  At every time step, the client  sits at location $a$. For each set $i$ of the set cover instance, the corresponding facility is in position $a$ if set~$i$ contains element~$t$, and in position $b$ otherwise. This defines the instance of dynamic facility location.

Since the distance is infinite between the two locations, the output $A$, to have  finite cost, has to correspond to a cover of the unique client for all  $T$ time steps, i.e. a cover of all $T$ elements in the set cover input. The cost is then simply $f$ times the number of selected facilities. The $\Omega(\ln T)$ hardness lower bound for Set Cover with $T$ elements in \cite{DinurSteurer2014STOC} implies that same lower bound for our problem.
\end{proof}

\section{Hourly opening cost}
\label{sec:hourly}

\subsection{Dynamic Facility Location with hourly opening cost}

We now focus on a variant of the problem studied in the previous section where the facilities can be open and closed at any time step and where the opening cost $f$ is paid for every facility open at every time step. 

\paragraph{Dynamic Facility Location problem with \emph{hourly} opening cost.} 
We are given a set $F$ of $m$ facilities and a set $C$ of $n$ clients together with a finite sequence of distances $(d_t)_{1\leq t \leq T}$ over $F\times C$, and two non-negative values $f$ and $g$. The goal is to output a sequence of subsets $A_t\subseteq F$ of facilities and, for each time step $t\in[T]$ an assignment, $\phi_t:C\rightarrow A_t$ of facilities to clients, so as to minimize:
\begin{equation*}
f\cdot \sum_{1\leq t\leq T}\#A_t
+\sum_{1\leq t\leq T, j\in C} d_t(\phi_t (j),j)
+g\cdot\sum_{1\leq t< T}\sum_{j\in C}\bigUn\{\phi_t(j)\neq\phi_{t+1}(j)\}.
\end{equation*}
 
\paragraph{Linear relaxation.}
The LP~(\ref{LP:1}) readily extends, with variables $y_i^t$:
\begin{equation}
\MinProg{     \displaystyle f\sum_{1\leq t\leq T} \sum_{i\in A} y^t_i 
			+ \sum_{j\in C}\sum_{1\leq t\leq T}\sum_{i\in F} x^t_{ij}\cdot d_t(i,j)
			+ g \sum_{j\in C} \sum_{1\leq t<T} \sum_{i\in F} z_{ij}^t 
}{&	\displaystyle(\forall ijt)~ x_{ij}^t\leq y^t_i\\[2mm]
&	\displaystyle(\forall jt)~ \sum_{i\in F} x_{ij}^t=1\\[5mm]
&	\displaystyle(\forall ij,~\forall t<T)~ z_{ij}^t\geq x_{ij}^t-x_{ij}^{t+1}\\[2mm]
& 	y^t_i,x_{ij}^t,z_{ij}^t\geq 0
}
\label{LP:2}
\end{equation}

\subsection{$O(\log nT)$-Approximation algorithm}

We now change the sampling procedure for the facilities: every facility~$i$ selects an exponentially distributed random threshold and opens only when its $y^t_i$ variable is above the threshold. 

\begin{algorithm}[H]
\caption{Hourly opening cost}\label{alg:hourly:alg}
\begin{algorithmic}
\sffamily
\STATE $\bullet$ Solve the  linear program LP (\ref{LP:2}). Let $(x,y,z)$ be the solution obtained.
\STATE $\bullet$ For each facility~$i$, pick a random threshold $\rho_i$ according to an exponential distribution with expectation $1/(2\log(2nT))$: i.e. $\Pr\{\rho_i>a\} = e^{-2a\log(2nT)}$ for all $a\geq0$. Open facility~$i$ at all times~$t$ such that $y_i^t>\rho_i$. 
For each time $t$, let $A_t$ be the resulting multiset of facilities open at  $t$.
\FOR{each client $j$}
	\STATE $\bullet$ Determine when it should change from one facility to another using the $z$-variables, and assign it to the cheapest selected facility between each change:
	\STATE ~~~~(a) As before, partition time greedily  into $\ell_j$ intervals $[t^j_k,t^j_{k+1})$ s.t. $t^j_1=1$ and where $t^j_{k+1}$ is inductively defined as the largest 
	$t\in(t^j_k,T+1]$ with ${
\displaystyle\sum_{i\in F} \Bigl(\min_{t^j_k\leq u < t}  x_{ij}^u\Bigr) \geq 1/2}$, and $t^j_{\ell_j+1}= T+1$;
	\STATE ~~~~(b) For each time interval $I=[t^j_k,t^j_{k+1})$ and facility $i$,  let ${x^I_{ij} = \min_{t\in I} x_{ij}^t}$ and connect client~$j$ to the facility~$i\in A_t$ that minimizes the ratio $\rho_i/x^I_{ij}$.
\ENDFOR
\end{algorithmic}
\end{algorithm}

The idea is that client~$j$ selects a facility~$i$ whose opening threshold is below its $x^I_{ij}$-value. We will show that one can find such a facility for all clients at all time steps with high probability.


Let us now analyze the cost of the resulting solution (Theorem~\ref{thm:hourly:alg}).

\begin{lem}
The expected opening cost is at most $2\log (2nT)$ times the corresponding term in LP~(\ref{LP:2}).
\end{lem}

\begin{proof}
Facility~$i$ is open at time~$t$ with probability $\Pr\{\rho_i \leq y^t_i\} = 1-e^{-2y^t_i\log(2nT)}\leq y^t_i\log(2nT)$ since $e^a \geq 1+a$ for all $a\in\mathbb R$. The expected facility cost is thus at most $f\sum_i\sum_t y^t_i \cdot 2\log(2nT)$.
\end{proof}

As before, Fact~\ref{fac:sumz>1/2}  holds here as well and the total switching cost is at most twice the value of  the corresponding in LP~(\ref{LP:2}). 

\begin{proof}[Proof of Theorem~\ref{thm:hourly:alg}]
We are now left with evaluating the distance cost. We want to show that we can view things so that there is a facility sampled according to $x^I_{ij}$ that is alive throughout the time interval $I$. Use the same $\rho_i$'s, but imagine that you only open facility~$i$ if $x^I_{ij}>\rho_i$. Since we assign client~$j$ to the facility~$i$ such that $\rho_i/x^I_{ij}$ is minimum, we do get a facility in that way, as long as $\rho_i/x^I_{ij} <1$. Note that $\rho_i/x^I_{ij}$ is an exponential of rate $2x^I_{ij}\log(2nT)$, and by independence of the $\rho_i$'s, $\min_i (\rho_i/x^I_{ij})$ is also an exponential of rate $2\sum_i x^I_{ij}\log(2nT)$, indeed:
$$
\Pr\{\min_i(\rho_i/x^I_{ij}) > a\} 
	= \prod_i \Pr\{\rho_i > a\cdot x^I_{ij}\} 
	= e^{-2a \cdot(\sum_i x^I_{ij}\log(2nT))}.
$$
Then, the probability that client~$i$ is not covered by an open facility by this process during time interval $I$ is $\Pr\{\min_i(\rho_i/x^I_{ij}) \geq 1\}  = e^{-2(\sum_i x^I_{ij})\log(2nT)}\leq \frac{1}{2nT}$ since $\sum_i x^I_{ij} \geq 1/2$. Consider now the event $B$ that for all client~$j$ and all interval $I$, $\min_i  (\rho_i/x^I_{ij}) < 1$, then $\Pr B \geq \frac12$ by the union bound. Now, conditioned to event $B$, the expected distance between every client $i$ to the facility it is assigned during each interval $I$  is upper bounded by: 
\begin{align*}
\sum_i \frac{\Pr\{\rho_i\leq x^I_{ij}\}}{\Pr B} d_t(i,j) 
&	\leq \sum_i 2\left(1-e^{-2x^I_{ij}\log(2nT)}\right) d_t(i,j)\\
&	\leq \sum_i 4x^I_{ij}\log(2nT) d_t(i,j)
\end{align*}
The expected distance cost conditioned to event $B$ is thus at most $4\log(2nT)$ times the corresponding term in LP~(\ref{LP:2}).

As $\Pr B \geq \frac12$, the expected facility cost conditioned to event $B$ is at most twice the unconditioned facility cost. The expected overall cost conditioned to event $B$ is then at most:
$$
4\log(2nT) f \sum_{i,t} y_i^t + 2 g \sum_{i,j,t} z_{ij}^t + 4\log(2nT) \sum_{i,j,t} x^I_{ij} d_t(i,j) \leq 4\log(2nT) \LP
$$ 
We conclude by applying Markov inequality as before.
\end{proof}

\section{Conclusion and open questions}

Algorithm~\ref{alg:fixed:alg} applies even if the distance do not follow the triangle inequality, and extends directly to  non-uniform opening cost as well as to arrival and departures dates for clients. It is striking that  instances with distances verifying the triangle inequality are not easier in the dynamic setting as opposed to the classic static setting (the approximation ratio $\Theta(\log nT)$ of Algorithm~\ref{alg:fixed:alg} is tight in both dynamic cases).  Algorithm~\ref{alg:hourly:alg} extends also directly to the setting of opening costs which are non-uniform in time as well. The last section raises naturally the question whether there is an $\omega(1)$-hardness result\,/\,$O(1)$-approximation algorithm for the general hourly opening cost case. We believe that our dynamic setting should be helpful in designing better \emph{static} representations of dynamical graphs (such as two dimensional flowcharts of the clients navigating between the different facilities over time). An other natural extension of our work is to study other objective functions for the distance cost, such as the sum of the diameters of resulting clusters at all time (i.e. the sum of the distance of the farthest client attached to each facility at all time, see e.g. \cite{CharikarPanigrahy2001} for a static formulation). As it turns out, the optimal dynamic solutions tend to adopt very intriguing behaviors under this objective, even in the simplest case of client moving along a fixed line, as has been observed in \cite{FernandesOshiroSchabanel2013Algotel}.

\bibliographystyle{plain}
\bibliography{dyn-clus}

\end{document}